\documentclass[conference]{IEEEtran}
\makeatletter
\def\ps@headings{%
\def\@oddhead{\mbox{}\scriptsize\rightmark \hfil \thepage}%
\def\@evenhead{\scriptsize\thepage \hfil \leftmark\mbox{}}%
\def\@oddfoot{}%
\def\@evenfoot{}}
\makeatother
\usepackage[T1]{fontenc}
\usepackage[latin9]{inputenc}
\usepackage[letterpaper]{geometry}
\usepackage{amsfonts}
\usepackage{graphicx,subfig}
\usepackage{amsmath,amsthm}
\usepackage{wrapfig}

\usepackage{epsfig}
\usepackage{amsmath}
\usepackage{amssymb}
\usepackage{psfrag}
\usepackage[absolute]{textpos}

\renewcommand{\baselinestretch}{0.9825}
\newcommand{\beq}{\begin{equation}}
\newcommand{\eeq}{\end{equation}}
\newcommand{\bea}{\begin{eqnarray}}
\newcommand{\eea}{\end{eqnarray}}

\long\def\symbolfootnote[#1]#2{\begingroup%
\def\thefootnote{\fnsymbol{footnote}}\footnote[#1]{#2}\endgroup}

\usepackage{amsfonts}
\usepackage{amsmath}
\usepackage{graphicx}
\usepackage{amssymb}

\def\real{{\mathchoice%
{\hbox{\rm\setbox1=\hbox{I}\copy1\kern-.45\wd1 R}}
{\hbox{\rm\setbox1=\hbox{I}\copy1\kern-.45\wd1 R}}
{\hbox{\scriptsize\rm\setbox1=\hbox{I}\copy1\kern-.45\wd1 R}}
{\hbox{\scriptsize\rm\setbox1=\hbox{I}\copy1\kern-.45\wd1 R}}}}
\def\Zint{{\mathchoice{\setbox1=\hbox{\sf Z}\copy1\kern-.75\wd1\box1}
{\setbox1=\hbox{\sf Z}\copy1\kern-.75\wd1\box1}
{\setbox1=\hbox{\scriptsize\sf Z}\copy1\kern-.75\wd1\box1}
{\setbox1=\hbox{\scriptsize\sf Z}\copy1\kern-.75\wd1\box1}}}
\newcommand{\complex}{ \hbox{\rm C\kern-0.45em\rule[.07em]{.02em}{.58em}%
\kern 0.43em}}

\newcommand{\be}{\begin{equation}}
\newcommand{\ee}{\end{equation}}
\newcommand{\beqr}{\begin{eqnarray}}
\newcommand{\eeqr}{\end{eqnarray}}
\newcommand{\beqrx}{\begin{eqnarray*}}
\newcommand{\eeqrx}{\end{eqnarray*}}
\newcommand{\ba}{\left[ \begin{array}}
\newcommand{\ea}{\\ \end{array} \right]}
\newcommand{\bi}{\begin{itemize}}
\newcommand{\ei}{\end{itemize}}

\newtheorem{lemma}{Lemma}
\newtheorem{theorem}{Theorem}

\newtheorem{definition}{Definition}

\newcommand{\x}{\mathbf{x}}
\newcommand{\y}{\mathbf{y}}
\newcommand{\w}{\mathbf{w}}

\begin{document}

\title{Compressive Sensing over Graphs}

\author{Weiyu Xu \ \ \ \ Enrique Mallada \ \ \ \ Ao Tang
\\
Cornell University, Ithaca, NY 14853, USA} \maketitle

\begin{abstract}
\boldmath
In this paper, motivated by network inference and tomography applications, we study the problem of compressive sensing for sparse signal vectors over graphs. In particular, we are interested in recovering sparse vectors representing the properties of the edges from a graph. Unlike existing compressive sensing results, the collective additive measurements we are allowed to take must follow connected paths over the underlying graph. For a sufficiently connected graph with $n$ nodes, it is shown that, using $O(k \log(n))$ path measurements, we are able to recover any $k$-sparse link vector (with no more than $k$ nonzero elements), even though the measurements have to follow the graph path constraints. We further show that the computationally efficient $\ell_1$ minimization can provide theoretical guarantees for inferring such $k$-sparse vectors with $O(k \log(n))$ path measurements from the graph.

\end{abstract}
\IEEEpeerreviewmaketitle

\section{Introduction}
  In operations of communication networks, we are often interested in inferring and monitoring the network performance characteristics,
  such as delay and packet loss rate, associated with each link. However, making direct measurements and monitoring for each link can be
   costly and operationally difficult, often requiring the participation from routers or potentially unreliable middle network nodes.
   Sometimes the responses from the middle network nodes are unavailable due to physical or protocol constraints.
   This raises the question of whether it is possible to quickly infer and monitor the network link characteristics from indirect
   end-to-end (aggregate) measurements. The problem falls in the area of \emph{network tomography}, which is useful for network traffic
   engineering \cite{zrwq} and fault diagnosis \cite{BGP}\cite{kleinberg}\cite{multihopwireless}\cite{Opticalprobabilistic}. 
  Because of its importance in practice, network tomography has seen a surge in excellent research activities performed from different angles, for example,  \cite{duffiedtomography}\cite{Bindelalgebraic}\cite{overviewtomography}\cite{Duffieldbinarynetworks} \cite{grouptestingongraphs}\cite{kleinberg}\cite{identifyingcode}\cite{endtoenddtawireless}\cite{binary}\cite{Bindelunbiased}. In this paper, we propose to study the basic network tomography problem from the angle of ``compressive sensing'', which aims to recover parsimonious signals from underdetermined or incomplete observations.

  Compressive sensing is a new paradigm in signal processing theory, which challenges to sample and recover parsimonious signals efficiently. It has seen quick acceptance in such applications as seismology, error correction and medical imaging since the breakthrough works \cite{CT1}\cite{BFP}\cite{C}\cite{DT}, although its role in networking is still limited \cite{Rabbat2}\cite{Rabbat3}\cite{Rabbat1}\cite{zrwq}. Its basic idea is that if  an object being measured is well-approximated by a lower dimensional object (e.g., sparse vector, low-rank matrix, etc.) in an appropriate space, one can exploit this property to achieve perfect recovery of the object. Compressive sensing \cite{CT1}\cite{C}\cite{DT} characterizes this phenomenon for sparse signal vectors, and presents efficient signal recovery schemes, from a small number of measurements. Recent works have started to extend this framework to the efficient inferring of low-rank matrices \cite{BFP}.

In this paper, we propose a compressive sensing approach for network (graph) tomography by exploiting the sparse signal structures therein. For example, it is very common that only a small fraction of network links are experiencing congestion or large packet loss rates. Compressive sensing appears to be the right tools to infer those sparse characteristics. However, many existing results of compressive sensing critically rely on assumptions that do not hold for network applications. For example, in network tomography, a measurement matrix is in a more restrictive class, taking only nonnegative integers while random Gaussian measurement matrices are commonly used in current compressive sensing literature. More importantly, as we will see, measurements are restricted by network topology and network operation constraints which are again absent in existing compressive sensing research. Overall, compressive sensing for network tomography, compared with other compressive sensing problems, is quite different and interesting in its own right because of its close connection to graphs. It is therefore not clear whether we have theoretical guarantees for recovering individual link characteristics using underdetermined observations under graph \emph{topology} constraints and if so, how to do it. This paper answers these two fundamental questions.

More concretely, bridging the gap between compressive sensing and graph theory, we study compressive sensing over graphs. The signal vectors to be recovered are sparse vectors representing the link parameters of a graph. We are allowed to take measurements following paths (walks) over the graph. We have the following two main results: for a sufficiently connected graph with $n$ nodes, even though under the \emph{graph path constraints},
\begin{itemize}
\item $O(k \log(n))$ path measurements are sufficient for identifying any $k$-sparse link vector (for example, identifying $k$ congested links)
\item $\ell_1$ minimization has a theoretical guarantee of recovering any $k$-sparse link vector with $O(k \log(n))$ path measurements.
\end{itemize}

 The paper is organized as follows. In Section \ref{subsec:delayexample}, we give the problem formulation, explain the special properties of compressive sensing over graphs, and compare it with graph constrained group testing problems. In Section \ref{sec:results}, we show that $O(k \log{(n)})$ path measurements are sufficient for compressive sensing over graphs.  In Section \ref{sec:algorithm}, we show that $\ell_1$ minimization can provably guarantee the performance of compressive sensing over graphs.  Section \ref{sec:numerical} presents numerical examples to confirm our predictions.  We conclude in Section \ref{sec:conclusion}.


\section{Problem Formulation and Related Works}
\label{subsec:delayexample}

We consider a network, represented by an undirected graph $G=(V,E)$, where $V$ is the vertex (or node) set with cardinality $|V|=n$, and $E$ is the edge (or link) set with cardinality $|E|$. Communications between vertices can only occur over these edges. Over each undirected edge between two vertices, communications can occur in both directions. \footnote{This undirected graph model has been used for communications networks such as optical networks  \cite{grouptestingongraphs}\cite{Opticalprobabilistic}. And Our work can also be extended to directed graph models. We also allow paths to visit an edge multiple times.} We also assume that each communication route must be a connected path over this undirected graph.

Suppose that we have probes along $m$ source-destination pairs over a network ($|E|>m$, otherwise the problem is not interesting). We are interested in identifying certain links from the probe measurements. For example, the congested links with large delays or high packet loss rates. We note that the delay over each source-destination pair is a sum of the delays over each edge on the route between this source-destination pair, giving a natural linear mixing of the link delays on the route. Abstractly, let $\x$ be an $|E| \times 1$ non-negative vector whose $j$-th element represents the delay (or $-\log(1-P_{j})$, where $P_{j}$ is the packet loss rate over link $j$) over  edge $j$ and let $\y$ be an $m \times 1$ dimensional vector whose $i$-th element is the end-to-end delay (or $-\log(1-P)$, where $P$ is the packet loss rate for the whole path) measurement for the $i$-th source-destination pair. Then
\begin{equation}
\y=A\x,
\end{equation}
where $A$ is an $m \times |E|$ matrix, whose element in the $i$-th row and $j$-th column is `1' if the $i$-th source-destination pair routes through the $j$-th link and `0' otherwise. For example, for a network with $|E|=6$ links and $m=4$ paths in Figure \ref{fig:example}, the measurement matrix $A$ is:

\begin{figure}
  \centering
  \includegraphics[width=0.2\textwidth]{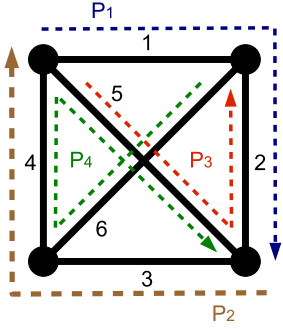}
   \caption{A Network Example}
  \label{fig:example}
\end{figure}

\begin{equation}
\label{eq:Aexample}
    A=\left( \begin{array}{cccccc}1&1&0&0&0&0 \\
    0&0&1&1&0&0 \\
    0&1&0&0&1&0 \\
    0&0&0&1&1&1
        \end{array}\right).
    \end{equation}

The question now is whether we can estimate the link vector $\x$, using the path measurement $\y$. Although $|E|>m$ means we only have an underdetermined system, it is still possible if we know $\x$ is a sparse vector, which in practice can often be a reasonable assumption. For example, there are only a small fraction of links that are congested, i.e., the link delays are considerably larger than the delays over other links. In other words, the vector $\x$ representing the delays over links is a spiky (or approximately sparse) vector. This provides the foundation to link our network tomography problems to compressive sensing. There are however important differences between network tomography problems and general compressive sensing formulation:

\begin{itemize}

\item Because of making measurements over communications paths, the element $A_{i,j}$ from $A$ is either $0$, when the measurement path $i$ does not go through link $j$, or an integer $b$, when the measurement path $i$ goes through link $j$ for $b>0$ times. Generally, the number $b$ is `1', which often makes the matrix a `0' and `1' matrix.

\item More importantly, besides being a `0-natural number' matrix, $A$ also has to satisfy the path constraints over the graph. Namely, all the nonzero elements in row $i$ of $A$ must correspond to a connected path. Even for a complete graph in Figure \ref{fig:example}, a row from $A$ can not take the form $(0,0,0,0,1,1)$. This is because no path can \emph{only} transverse link $5$ and link $6$.


\item In many cases, the sparse link vectors we are interested in are nonnegative vectors. For instance, the delay vectors and the inverse logarithm of the packet loss rate vector.

\end{itemize}

Finally, we want to compare our study with a closely related topic, graph-constrained group testing \cite{monitoringcycle,grouptesting, groupbook, grouptestingongraphs,binary}. Compressive sensing over graphs involves $\y$ which can take values over real numbers, instead of `true-or-false' binary values for the group testing problems. The measurement result $\y$ is the additive linear mixing of the vector $\x$ over real numbers, in contrary to the logic OR operation for group testing problems. Consider a simple example, if the delay vector $\x$ for the network in Figure \ref{fig:example} is $(2,3,0,0,0,0)^{T}$, then in compressive sensing, $\y=(5,0,3,0)^{T}$; while in group testing, $\y=(\text{Y}, \text{N},\text{Y}, \text{N})$, where $\text{Y}$ and $\text{N}$ represent ``congested" and ``not congested" respectively.  From compressive sensing, by a simple checking, we know $\x=(2,3,0,0,0,0)^{T}$ is the only sparsest solution  that satisfies $\y=A\x$; however, group testing will decide that $\x=(\text{N},\text{Y},\text{N},\text{N},\text{N},\text{N})^{T}$. But in fact, there is no $1$-sparse $\x$ that can generate such a $\y=(5,0,3,0)^{T}$. This hints that compressive sensing can do better than group testing in terms of needed measurements which will be further quantified in table I.

\section{When is Compressive Sensing over Graphs Possible?}
\label{sec:results}

In this section, we focus on the question that how many path observations will suffice to recover any $k$ network edge failure. First, in an order of more and more demanding requirements, we give three conditions on the measurement matrix $A$ to guarantee recovering $k$-sparse link vectors (Theorems \ref{thm:nonnegative}, \ref{thm:rankcondition} and \ref{thm:special_condition}). Then we show that a measurement matrix generated from random walks will be able to recover any $k$-sparse vector using only $O(k\log(n))$ measurements.

\subsection{Success Conditions for Compressive Sensing}
\begin{theorem}
Let $\y=A\x$. Then if $\x$ is a nonnegative signal vector with no more than $k$ nonzero elements, with
 \begin{equation*}
 k < \min_{\w \in \mathcal{N}(A), \w \neq {0}}\max\{k_{-,\w},k_{+,\w}\},
 \end{equation*}
where $\mathcal{N}(A)$ is the null space of $A$, $k_{-,\w}$ and $k_{+,\w}$ are the number of negative and positive nonzero elements in the vector $\w$, then any such nonnegative signal vector is the unique sparsest nonnegative vector satisfying $\y=A\x$. Conversely, if
\begin{equation*}
 k \geq  \min_{\w \in \mathcal{N}(A), \w \neq 0}\max\{k_{-,\w},k_{+,\w}\},
 \end{equation*}
then there exists a nonnegative $k$-sparse vector $\x$ such that it is not the unique sparsest nonnegative vector satisfying $\y=A\x$.
\label{thm:nonnegative}
\end{theorem}

\begin{proof} We first prove the forward direction. Indeed, any vector $\tilde{\x}$ satisfying $\y=A\tilde{\x}$ must be of the form $\tilde{\x}=\x+\w$ with $\w$ from the null space of the matrix $A$. If $k <k_{-,\w}$, then the $k$-sparse nonnegative vector $\x$ plus the vector $\w$ will have at least one negative element, which can not be a nonnegative solution to $\y=A\x$. If instead $k <k_{+,\w}$, then the $k$-sparse nonnegative vector $\x$ plus the vector $\w$ will have at least $k_{+,\w}$ nonzero elements, which must have more than $k$ nonzero elements.

Now we only need to prove that we can always find a $k$-sparse signal $\x$ with  $ k \geq \min_{\w \in \mathcal{N}(A), \w \neq 0}\max\{k_{-,\w},k_{+,\w}\}$ such that $\x$ is not the unique sparsest solution satisfying $\y=A\x$. We let $\w \in \mathcal{N}(A)$ denote the nonzero vector minimizing $\max\{k_{-,\w},k_{+,\w}\}$.

In fact, if we take a vector $\x$ supported on the set $K$, with $|K|=k=\max\{k_{-,\w},k_{+,\w}\}$, $K_{-,\w} \subseteq K$ and $K \subseteq K_{-,\w} \bigcup K_{+,\w}$, where $K_{-,\w}$ is the index set for the negative elements of $\w$ and $K_{+,\w}$ is the index set for the positive elements of $\w$.

We let $\x_{K}=|\w_{K}|$ (taking elementwise absolute value). Then obviously, $\x+\w$ will be a $k_{+,\w}$-sparse nonzero vector, and has no more than $k$ nonzero elements.

\end{proof}

For comparison, we have a more stricter, but easier to use condition for recovering an arbitrary (not necessarily nonnegative) $k$-sparse vector $\x$.

\begin{theorem}
Let $\y=A\x$.  If $\x$ is a signal vector with no more than $k$ nonzero elements, where
 \begin{equation*}
 k < \min_{\w \in \mathcal{N}(A), \w \neq 0}\frac{\|\w\|_{0}}{2},
 \end{equation*}
where $\|\w\|_{0}$ is the number of nonzero elements in the vector $\w$, then $\x$ is the unique sparsest vector satisfying $\y=A\x$. Conversely, if
 \begin{equation*}
 k \geq  \min_{\w \in \mathcal{N}(A), \w \neq 0}\frac{\|\w\|_{0}}{2},
 \end{equation*}
then there exists a $k$-sparse vector $\x$ such that it is not the unique sparsest vector satisfying $\y=A\x$.
\label{thm:rankcondition}
\end{theorem}

\begin{proof}
Following the same line of proof in Theorem \ref{thm:nonnegative}.
\end{proof}

Based on the previous theorems, we can now give a stricter sufficient condition for recovering $k$-sparse signal.

\begin{theorem}
\label{thm:special_condition}
Suppose that for every no more than $h$ columns, indexed by the set $H \subseteq \{1,2, ..., |E|\}$, of the $m \times |E|$ measurement matrix $A$, the corresponding $m \times h$ submatrix $A_{H}$ (consisting of these $h$ columns of $A$) has at least one row, say row $i$, such that there is a single nonzero element in that row. Then any $k$-sparse signal vector $\x$, with $k < \frac{h+1}{2}$, is the unique sparsest solution $\x$ to $\y=A\x$.

\end{theorem}

\begin{proof}
From Theorem \ref{thm:rankcondition}, we only need to show that in the null space of $A$, every nonzero vector will have at least $(h+1)$ nonzero elements. In fact, suppose that there exists a nonzero vector $\w \neq 0$ from the null space of $A$, which has no more than $h$ nonzero elements, and suppose that its support set is $H$. However, since there exists one row in $A_{H}$ with a single nonzero element, $A_{H}\w_{H}$ must be nonzero, which contradicts the fact that $\w$ is from the null space of $A$. So each nonzero vector in the null space of $A$ has at least $(h+1)$ nonzero elements. From Theorem \ref{thm:rankcondition}, every $k$-sparse vector $\x$, with $k < \frac{h+1}{2}$, will be the unique sparsest solution to $\y=A\x$.

\end{proof}

\subsection{How Many Measurement Paths are Needed?}
Now we want to show that $O(k\log(n))$ measurements are enough for recovering any $k$-sparse link vector for a sufficiently connected graph with $n$ nodes.

\subsubsection{Graph Assumptions}
Before we proceed, following the works on graph-constrained group testing \cite{grouptestingongraphs,grouptesting}, we introduce the following assumptions on the graphs.

  The undirected graph $G=(V,E)$ is called a $(D,c)$ uniform graph if for some constant $c$, the degree of each vertex $v \in V$ is between $D$ and $cD$. Suppose that a standard random walk over the graph has a stationary distribution $\mu$ over the nodes. The $\delta$-mixing time of $G$ is defined as the smallest $t'$ such that a random walk of length $t'$ starting at any vertex in $G$ ends up having a distribution $\mu'$ such that $\|\mu-\mu'\|_{\infty} \leq \delta$. We define $T(n)$ as the $\delta$-mixing time of $G$ for $\delta=\frac{1}{(2cn)^2}$.

\subsubsection{$O(k \log(n))$ measurements are sufficient}

In compressive sensing, we adopt an $m \times |E|$ measurement matrix generated by $m$ independent random walks . For each random walk, we uniformly randomly pick a starting vertex from $V$ and then perform a standard random walk over the graph. The length of the random walk is denoted by $t$.
From \cite{grouptesting}, we have the following theorem,
\begin{theorem}
\cite{grouptesting}
There is a degree $D_{0}=O(c^2kT^2(n))$ and $t=O(\frac{nD}{c^3kT(n)})$ such that whenever $D \geq D_{0}$, by setting the path lengths $t=O(\frac{nD}{c^3 kT(n)})$ the following holds. Let $B$ be a set of at most $(k-1)$ edges in the graph $G$, and let $e$ be an edge not belonging to the set $B$.  Then
\begin{equation*}
\pi_{e,B}=\Omega(\frac{1}{c^4 k T^2(n)}),
\end{equation*}
where $\pi_{e,B}$ is the probability that the random walk passes through link $e$, but misses all the edges from the set $B$.
\end{theorem}

Now we take an arbitrary set of edges $E'$ with cardinality $|E'|=k$. Let us take $m$ independent measurements satisfying the graph path constraints. Then the probability that there does not exist any measurement walk (each walk corresponds to a row of the measurement matrix $A$) with a single nonzero element in the columns corresponding to the edges from $E'$, can be expressed by
\begin{equation*}
   P=(1-\pi_{E'})^{m},
\end{equation*}
where $\pi_{E'}$ is the probability that a random walk visits one and only one element from the set $E'$. In fact, $\pi_{E'}=\Omega(\frac{1}{c^4 k T^2(n)})\times k$ since the events of having a single nonzero element can be divided into $k$ disjoint events, each of which is the event that the single nonzero element appears in one of the $k$ possible columns of $E'$.

Since there are $\binom{|E|}{k}$ ways of choosing the $k$ edges, the probability that there exists one edge set $E'$ of $|E'|=k$ without any single-nonzero-element row, is

\begin{eqnarray}
P_{k,k}&\leq&\binom{|E|}{k} (1-\pi_{E'})^{m}\\
       &\leq& \binom{n^2}{k} (1-\Omega(\frac{k}{c^4 k T^2(n)}))^m\\
       &\leq& e^{k(1+\log(\frac{n^2}{k}))+m\log(1-\Omega(\frac{1}{c^4T^2(n)}))}
\end{eqnarray}

So if
\begin{equation*}
e^{k(1+\log(\frac{n^2}{k}))+m\log(1-\Omega(\frac{1}{c^4T^2(n)}))}<1,
\end{equation*}
namely
\begin{equation*}
m > -\frac{{k(1+\log(\frac{n^2}{k}))}}{\log(1-\Omega(\frac{1}{c^4T^2(n)}))},
\end{equation*}
the probability $P_{k,k}$ will be smaller than $1$.

Now let us look at a set $E''$ with cardinality $|E''|=k_1$ smaller than $k$.
We notice that $\pi_{e,B}=\Omega(\frac{1}{c^4 k T^2(n)})$ is true
for any edge $e$ and any set $B$ of cardinality no bigger than $k$.  So the probability $\pi_{E''}$ that a random walk visits edge $e$ (and only visits that edge $e$) from the set $E''$ is $\pi_{e, E''\setminus e}= \Omega(\frac{1}{c^4 k T^2(n)})$.

Again we take $m$ independent random walk  measurements satisfying the graph constraints. Then the probability that there does not exist any measurement having one and only one nonzero element in the columns corresponding to the edge set $E''$ is given by

\begin{equation*}
   P=(1-\pi_{E''})^{m},
\end{equation*}
where $\pi_{E''} =\Omega(\frac{k_1}{c^4 k T^2(n)})$ since the events of having a unique nonzero element over $k_1$ different columns are disjoint events.

Since there are $\binom{|E|}{k_1}$ ways of choosing the $k_1$ edges, the probability that there exists one edge set $E'$ with $|E'|=k_1$ without any desired single-nonzero-element row is

\begin{eqnarray}
P_{k_1,k}&\leq&\binom{|E|}{k_1} (1-\pi_{E''})^{m}\\
       &\leq& \binom{n^2}{k_1} (1-\Omega(\frac{k_1}{c^4 k T^2(n)}))^m\\
       &\leq& e^{k_1(1+\log(\frac{n^2}{k_1}))+m\log(1-\Omega(\frac{k_1}{c^4 k T^2(n)}))}
\end{eqnarray}

So if
\begin{equation*}
k_1(1+\log(\frac{n^2}{k_1}))+m\log(1-\Omega(\frac{k_1}{c^4 k T^2(n)}))<0,
\end{equation*}

namely
\begin{equation*}
m > -\frac{{k_1(1+\log(\frac{n^2}{k_1}))}}{\log(1-\Omega(\frac{k_1}{c^4kT^2(n)}))}.
\end{equation*}

So as long as $m > \max_{1 \leq k_1 \leq k}{-\frac{{k_1(1+\log(\frac{n^2}{k_1}))}}{\log(1-\Omega(\frac{k_1}{c^4kT^2(n)}))}}$, with probability $1-o(1)$, the measurement matrix $A$ guarantees recovering up to $\frac{k}{2}$-sparse link vectors (from Theorem \ref{thm:special_condition}). In fact, $m=O(c^4T^2(n)k\log(n))$ measurement paths suffice.

The following table provides a summary of results for number of measurements needed in graph constrained problems or general problems without graph constraints.

\begin{table}
\label{table:comparison}
\centering
\begin{tabular}{|c|c|c|}
  \hline
   $m$& Compressive sensing & Group Testing \\
   \hline
  Graph constrained & $O(k\log(n))$(this paper) & $O(k^2\log(\frac{n}{k})) \cite{grouptesting}$ \\
  \hline
  General & $O(k \log(\frac{n}{k}))\cite{C}$ & $O(k^2\log(\frac{n}{k}))\cite{groupbook}$ \\
  \hline
\end{tabular}
\caption{Number of measurements needed in different scenarios}
\end{table}

\section{$\ell_1$ Minimization Decoding}
\label{sec:algorithm}
$\ell_1$ minimization has been a popular efficient decoding method for inferring $\x$ from compressed measurements $\y=A\x$ \cite{C,DT}. $\ell_1$ minimization solves for $\min{\|\x\|_1}$ subject to the constraint $\y=A\x$. However, it is not clear how one can \emph{efficiently} infer these sparse vectors over \emph{graphs}. In this section, we show that when the number of measurement paths is $m=O(k \log(n))$,  $\ell_1$ minimization can recover any $k$-sparse link vector efficiently. We will consider the matrix $A$ generated by \emph{regularized} random walks with ``\emph{good starts}''. Our proof strategy is to show that under the same graph assumptions as in last section, $A$ corresponds to a bipartite \emph{expander} graph with high probability. Then we use the expansion property to show the null space property of $A$ guarantees the success of $\ell_1$ minimization. Theorem \ref{thm:conditional} states that if the random walk ever visits a small edge set, very likely it visits this set a small number of times. Based on Theorem \ref{thm:conditional}, Theorem \ref{thm:ksetconcentration} and Theorem \ref{thm:walkexpansionforsmallk_1} assert that $A$ corresponds to an bipartite expander graph. Theorem \ref{thm:sing_col_concentration} and Lemma \ref{lem:maxdegreeof2} give further regularity properties of $A$. Finally, Theorem \ref{thm:l1works} shows how expansion property implies that $\ell_1$ minimization succeeds in recovering sparse vectors.

We first give the definitions about ``good start" random walks, measurement matrix $A$ constructed from regularized random walks, bipartite graphs corresponding to $A$ and some basic assumptions about the graph we are considering.

\begin{definition}[``good start'' random walk]
\label{def:proprandom}
A random walk with a ``good start" chooses the starting vertex with a probability proportional to its degree and then performs a random walk of length $t$ over the graph. Namely, the probability that the random walk starts with the vertex $i$ with probability $\frac{d_{i}}{2|E|}$, where $d_i$ is the degree of vertex $i$ and $|E|$ is the total number of edges in the graph.
\end{definition}

\begin{definition}[matrix $A$ from regularized random walks]
\label{def:Afromproprandom}
Suppose $W_{2}$ is a walk on an undirected graph $G=(V,E)$. Then a regularized walk $W_{1}$ adapted from $W_{2}$ is a walk which visits
the same set of edges as $W_{1}$ does, but visits each such edge no more than twice. We will use the regularized walks adapted from ``good start'' random walks to construct the rows of $A$.
\end{definition}
From Lemma \ref{lem:maxdegreeof2}, we can always get a \emph{regularized} walk from a given walk. However, using regularized walks, the maximum element in $A$ is upper bounded by $2$.

\begin{definition}[bipartite graph from an $m \times |E|$ matrix $A$]
\label{def:expander}
 We construct a bipartite graph by placing $|E|$ ``edge'' nodes on the left-hand side and $m$ ``measurement'' nodes on the righthand side. An ``edge'' node $j$ on the left is connected to a ``measurement'' node $i$ on the right if and only if the $i$-th random walk goes through edge $j$. For $0<\epsilon<1$, a bipartite graph is called a $(k, \epsilon)$  expander if every set of left nodes $S$, with cardinality $|S|\leq k$, are connected to at least $(1-\epsilon)|E(S)|$ righthand side nodes (namely the neighbors of $S$, denoted by $N(S)$), where $E(S)$ is the set of links that go from $S$ to the righthand side. In other words, $|E(S)|$ is the total number of nonzero elements in the columns corresponding to $S$ in $A$, $|N(S)|$ is the number of nonzero rows in the submatrix $A_{S}$ and $N(S)\geq (1-\epsilon)|E(S)|$. $d_{min}$ and $d_{max}$ are respectively the smallest and largest degrees of the left-hand ``edge'' nodes in the bipartite graph.
\end{definition}
For example, Figure \ref{fig:bipartite} is the corresponding bipartite graph for matrix $A$ in (\ref{eq:Aexample}).

\begin{figure}
  \centering
  \includegraphics[width=0.2\textwidth]{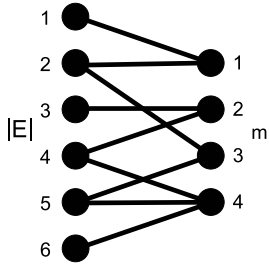}
   \caption{A Bipartite Graph Representation for $A$}
  \label{fig:bipartite}
\end{figure}

In \emph{this} section, we set $t=O(\frac{|E|}{k})$ and also assume the mixing time $T(n)$ has an upper bound as $n$ grows, which will simplify the presentation of our analysis. However, our results still extend to the case of growing $T(n)$ by setting $t=O(\frac{|E|}{T(n)k})$ and $m=O(T(n)^2k \log(n))$. We also assume that the smallest degree $D$ in the graph grows with $n$.

To prove the expansion property for $A$, for an arbitrary edge set $S$ with $|S|=k$, we bound the conditional probability that a random walk visits another edge in $S$ after it has already visited  one edge from $S$.

\begin{theorem}
\label{thm:conditional}
Let $P_{\geq 1,S}$ be the probability that a ``good start'' random walk ever visits an edge from an edge set $S$ with $|S|=k$. Let $P_{\geq 2,S}$ be the probability that such a random walk visits at least two edges from $S$. Then we can always select the random walk length in such a way that $t=O(\frac{|E|}{k})$ and $P_{\geq 2,S}\leq \eta P_{\geq 1,S}$, namely the conditional probability
\begin{eqnarray}
\label{eq:conditional}
&&P(\text{the random walk visits} > 1 ~~\text{edges in}~S|\\ \nonumber
&&\text{a random walk visits at least}~1~\text{edge in}~S)\\ \nonumber
&&\leq \eta,
\end{eqnarray}
where $0<\eta<1$ is a constant which can be made arbitrarily close to $0$. Similarly, for any $1<k'<k$, $P_{\geq (k'+1),S}\leq \eta P_{\geq k',S}$, where $P_{\geq k',S}$ $(P_{\geq (k'+1),S})$ is the probability that the random walk visits at least $k'$  ($k'+1$) edges from $S$, and $\eta$ is the same $\eta$ as in (\ref{eq:conditional}).
\end{theorem}

\begin{proof}
\label{proofthmconditional}
Suppose that the random walk ever visits one or more edges from the set $S$ and suppose the first edge from $S$ the random walk visits is edge $i \in S$, visited between time indices $j-1$ and $j$, where $1\leq j \leq t$. By denoting the two vertices connected by edge $i$ as $v_{i,1}$ and $v_{i,2}$, we also assume that at time index $j$, the random walks is at the $l$-th ($l=1~\text{or}~2$) vertex, denoted by $v_{i,l}$, of edge $i$. We denote the probability of this event by $P_{i,v_{i,l},j}$, $(i \in S)$, and further denote by $P_{\geq 2| i,v_{i,l},j}$ $(i \in S)$ the conditional probability that the random walk visits another edge from $S$ conditioned on this event (the random walk visits $i \in S$ first between time index $j-1$ and $j$, and sits at vertex $v_{i,l}$ at time index $j$).

Since the probability $P_{\geq 1, S}$ that the random walk visits at least one edge in $S$ can be decomposed as
\begin{equation*}
P_{\geq 1, S}={\sum_{i \in S} \sum_{j} \sum_{l=1}^2}P_{i,v_{i,l},j}~,
\end{equation*}
we have
\begin{equation*}
P_{\geq 2|\geq 1,S}=\frac{\sum_{i \in S} \sum_{j} \sum_{l=1}^2P_{i,v_{i,l},j} \times P_{\geq 2|i,v_{i,l},j}} {{\sum_{i \in S} \sum_{j} \sum_{l=1}^2}P_{i,v_{i,l},j}},
\end{equation*}
where $P_{\geq 2|\geq 1,S}$ is the conditional probability that the random walk visits at least one more edge in $S$ after already visiting one edge in $S$.

Now if we can show the conditional probability $P_{\geq 2|i, v_{i,l}, j}$ is small enough for every possible $i$, $j$ and $l$,  we will get the conclusion in the theorem. By the Markov property of the defined random walk, $P_{\geq 2| i,v_{i,l},j}$ $(i \in S)$  is upper bounded by the conditional probability that the random walk visits at least one edge of $S$ after time index $j$, conditioned on that the walk sits at the vertex $v_{i,l}$ at time index $j$.

So we only need to show that
\begin{eqnarray*}
&&P(\text{the random walk visits $S$ again after time index $j$}\\ \nonumber
&&|\text{the random walk is at vertex $v_{i,l}$ at time index $j$}), \nonumber
\end{eqnarray*}
is small enough or can be made arbitrarily close to $0$ if we choose the length of the random walk appropriately.
Before we proceed to upper bound this probability, we present the following lemma about the conditional probability that a random walk visits a certain edge after the mixing time $T(n)$.

\begin{lemma}
\label{lem:p_revisiting}
For any vertex $v_{i,l}$ and any time index $j$, if $z \geq T(n)$ (the $\delta$-mixing time), the conditional probability $P_{j+z, e| v_{i,l},j}$ that the random walk visits one certain edge $e$ between time index $j+z$ and $j+z+1$ is between $\frac{1}{|E|}-\frac{2\delta}{D}$ and $\frac{1}{|E|}+\frac{2\delta}{D}$, where $D$ is the smallest degree for the vertices in the graph.
\end{lemma}

\begin{proof}
At time index $j+z$, no matter what vertex the random walk is at time index $j$, by the definition of mixing time, the random walk will visit the two vertices that define edge $e$ with probabilities in the regions$[\frac{d_{e}^{1}}{2|E|}-\delta, \frac{d_{e}^{1}}{2|E|}+\delta]$ and $[\frac{d_{e}^{2}}{2|E|}-\delta, \frac{d_{e}^{2}}{2|E|}+\delta]$ respectively. So between time index $j+z$ and $j+z+1$, the probability that the random walk visits edge $e$ will be lower bounded by
\begin{eqnarray}
&~&(\frac{d_{e}^{1}}{2|E|}-\delta) \times \frac{1}{d_{e}^1}+(\frac{d_{e}^{2}}{2|E|}-\delta) \times \frac{1}{d_{e}^2}\\\nonumber
&=& \frac{1}{|E|}-\frac{\delta}{d_{e}^{1}}-\frac{\delta}{d_{e}^{2}}\geq \frac{1}{|E|}-\frac{2\delta}{D},\nonumber
\end{eqnarray}
and similarly, we have the upper bound.
\end{proof}

 Building on Lemma \ref{lem:p_revisiting}, to get the probability that the random walk visits another edge from $S$ conditioned on the fact it sits at node $v_{i,l}$ at time $j$, we divide the random walk after time index $j$ into $T(n)$ edge chains $f_1, f_2, f_3,...,f_{T(n)}$ constructed in the following way. The $s$-th chain $f_{s+1}$ ($0\leq s \leq T(n)-1$) starts from the edge traversed by the random walk between time indices $j+s$ and $j+s+1$. Then the $s$-th chain will include sequentially the edges traversed by the random walk between time index pairs $(j+s+T(n), j+s+1+T(n))$, $(j+s+2T(n), j+s+1+2T(n))$, ..., until the random walk ends. Namely, we sample the random walk (after time $j$) with a period of $T(n)$ with $T(n)$ different starting phases.

Without loss of generality, we look at a chain $f_{s+1}$. At time index $j+s$, the conditional probability (conditioned on the fact the random walk is at vertex $v_{i,l}$ at time index $j$) that the next edge traversed by $f_{s+1}$ is from $S$ is at most $\frac{k}{D}$, because no matter what vertex the random walk reached at time index $j+s$, there are at least $D$ edges connected to that vertex. Now we look at the probability $P_{s}$ that $f_{s+1}$ does not traverse any edge from $S$ after time index $s+j+T(n)$ (conditioned on the fact the random walk is at vertex $v_{i,l}$ at time index $j$). Since all the time indices are separated from each other and from vertex $v_{i,l}$ by at least $T(n)$ time slots, from the mixing time definition, $P_{s}$ is at least
$\left(1-(\frac{k}{|E|}+\frac{2k\delta}{D})\right)^{\lceil\frac{t}{T(n)} \rceil}$, where $\lceil \cdot \rceil$ represents the ceiling operation.

So the (conditional) probability that $f_{s}$ visits $S$ after time index $j$ is upper bounded by
\begin{equation*}
\frac{k}{D}+1-\left(1-(\frac{k}{|E|}+\frac{2k\delta}{D})\right)^{\lceil\frac{t}{T(n)} \rceil}   .
\end{equation*}

Using  a union bound over the $T(n)$ chains, the conditional probability that the random walk visits $S$ again will be upper bounded by
\begin{equation*}
\frac{kT(n)}{D}+T(n)\left( 1-\left(1-(\frac{k}{|E|}+\frac{2k\delta}{D})\right)^{\lceil\frac{t}{T(n)} \rceil}   \right),
\end{equation*}
which can be further upper bounded by
\begin{eqnarray*}
&&\frac{kT(n)}{D}+T(n)\times \left( \frac{k}{|E|}+\frac{2k\delta}{D}\right)\times {\lceil\frac{t}{T(n)} \rceil} \\
&\leq& \frac{kT(n)}{D}+\frac{(t+T(n))k}{|E|}+\frac{2k(t+T(n))\delta}{D}.
\end{eqnarray*}
So we can always take $t$ scaling as $O(\frac{|E|}{k})$ to make this probability arbitrarily small (of course $k$ must also make the first and third term small enough, which is easily true based on the assumptions on $D$, $T(n)$ and $\delta=\frac{1}{(2cn)^2}$ ).

Moreover, using the same set of arguments, we can extend this conclusion to any $1<k'<k$.
\end{proof}
\begin{theorem}
\label{thm:ksetconcentration}
With $t=O(\frac{|E|}{k})$, for any arbitrary edge set $S$ with cardinality $k$, if we take $m=O(T(n)k \log(n))$ ``good start'' random walks, then with probability $1-O(|E|^{-k})$, the number of walks that traverse at least one edge of $S$ is $g=\Theta(k\log(n))$; moreover, with probability $1-O(|E|^{-k})$, the total sum number of edges from $S$ visited by the $m$ random walks will be upper bounded by $r=(1+\epsilon')\frac{g}{1-\eta}$, where $\eta$ is the conditional probability appearing in Theorem \ref{thm:conditional} and $\epsilon'$ is an arbitrarily small number.
\end{theorem}

\begin{proof}
We start by providing a lower bound on the probability that the random walk ever visits $S$.
\begin{lemma}
\label{lem:setSprobability}
The probability $P_{\geq 1}$ that a random walk of length $t$ visits an edge set $S$ of cardinality $k$ will be $\Omega(\frac{tk}{T(n)|E|})$.
\end{lemma}

\begin{proof}
We consider a chain of period $T(n)$ and focus on the time slots starting with time index $0$, $T(n)$, $2T(n)$.... Note that at time index $0$, the random walk has achieved its stationary distribution due to the manner by which we pick the starting vertex. Similar to the proof of Theorem \ref{thm:conditional}, from the Markov property and the mixing time definition, the probability $P_{\geq 1}$ that the random  walk visits an edge in $S$ is lower bounded by
\begin{eqnarray*}
&&1-\left( 1-\left(\frac{k}{|E|}-\frac{2k\delta}{D}\right)\right)^{   \lfloor \frac{t}{T(n)} \rfloor} \\
&\geq& 1-e^{\lfloor \frac{t}{T(n)} \rfloor \log\left(1-\left(\frac{k}{|E|}-\frac{2k\delta}{D} \right)\right)}\\
&\geq& 1-e^{-\lfloor \frac{t}{T(n)} \rfloor \left(\frac{k}{|E|}-\frac{2k\delta}{D} \right)}= \Omega(\frac{t k}{|E|T(n)}).
\end{eqnarray*}
\end{proof}

Let $X_{i}$, $1 \leq i \leq m$, be $m$ independent Bernoulli random variables indicating whether the $i$-th random walk visits the set $S$,  so each of them takes value `1' with probability $P_{\geq 1}$ and takes value `0' with probability $(1-P_{\geq 1})$. Let $X=\sum_{i=1}^{m}X_{i}$ be the total number of walks that visit the set $S$. When $t=O(\frac{|E|}{k})$ and $m=O(T(n)k \log(n))$, the expected value of $X$ is $P_{\geq 1}m=P_{\geq 1} O(T(n)k \log(n))$. Now we show that the actual number of random walks that visit $S$ concentrates around $\Theta(k\log(n))$.

From a Chernoff bound on $X$, the probability that $X \geq P'm$ when $P'\geq P_{\geq 1}$ (or $X \leq P'm$ when $P' \leq P_{\geq 1}$) is upper bounded by $e^{-m Diff(P'||P_{\geq 1})}$,
where $Diff(P'||P_{\geq 1})$ is the relative entropy
\begin{equation*}
P'\log\left(\frac{P'}{P_{\geq 1}}\right)+(1-P')\log\left(\frac{1-P'}{1-P_{\geq 1}}\right).
\end{equation*}

So as long as
\begin{equation*}
m \geq  \frac{k \log(|E|)}{Diff(P'||P_{\geq 1})},
\end{equation*}
with probability $1-O(|E|^{-k})$, $X$ will concentrate around its mean value $mP_{\geq 1}$ (not going above or below $mP'$). If $P'=(1-\epsilon')P_{\geq 1}$ or $P'=(1+\epsilon')P_{\geq 1}$ for a sufficiently small $\epsilon'>0$,
\begin{equation*}
Diff(P'||P_{\geq 1})\thickapprox \frac{{\epsilon'}^2P_{\geq 1}}{1-P_{\geq 1}}.
\end{equation*}
So from Lemma \ref{lem:setSprobability}, when $m=O(T(n)k\log(n))$, with probability $1-O(|E|^{-k})$, the number of non-all-zero rows will be $g=\Theta(k\log(n))$.

Now let $Y_{i}$, $1\leq i\leq m$, be $m$ independent random variables indicating how many edges from $S$ the $i$-th random walk visits.  Let $Y=\sum_{i=1}^{m}Y_{i}$ be the total number of edges from $S$ visited by $m$ independent random walks.  From Theorem \ref{thm:conditional}, when $m=O(T(n)k \log(n))$, then the probability that $Y \geq r=(1+\epsilon')\frac{g}{1-\eta}$ will be no bigger than the probability $Y'= \sum_{i=1}^{m}Y_{i}'\geq r$, where $Y_{i}'$s are i.i.d. nonnegative integer-valued random variables and each of these $m$ random variables takes value `0' with probability $1-P_{\geq 1}$, `1' with probability $P_{\geq 1}(\eta-\eta^2)$, value `2' with probability $P_{\geq 1}(\eta^2-\eta^3)$,... and so on. So for each $1\leq i\leq m$, $\mathbb{E}(Y_{i}')=\frac{P_{\geq 1}}{1-\eta}$, and $\mathbb{E}(Y')=\frac{mP_{\geq 1}}{1-\eta}$. For any $\epsilon'>0$, by a standard Chernoff bound for $Y'$, with $m=O(T(n)k\log(n))$, $Y \geq (1+\epsilon')\frac{mP_{\geq 1}}{1-\eta}$ with probability at most $O(|E|^{-k})$. (We however choose not to present the explicit large deviation exponent for $Y$ in this paper due to its complicated expression.)
\end{proof}

Since there are at most $\binom{|E|}{k}$ edge sets of cardinality $k$, by a union bound and Theorem \ref{thm:ksetconcentration} (where we replace $n$ with $|E|\leq n^2$), with probability $1-o(1)$, for all the edge sets $S$ with cardinality $k$, the number $N(S)$ of random walks that visit $S$ will be at least $\frac{1-\eta}{1+\epsilon'}|E(S)|$. Note that this corresponds to the expansion concept we mentioned at the beginning of this section.

By repeating the previous arguments for smaller edge sets, we know with high probability, the expansion properties for all the edge sets with cardinality $\leq k$ also hold when $m=O(T(n)k\log(n))$. So in the end, we have the following theorem about expansion.
\begin{theorem}
\label{thm:walkexpansionforsmallk_1}
 If $t=O(\frac{|E|}{k})$, then a measurement matrix generated by $m=O(T(n)k \log(n))$ ``good start" random walks with length $t$ will be
an $(k, 1-\frac{1-\eta}{1+\epsilon'})$ expander, where $\eta$ is the same $\eta$ appearing in Theorem \ref{thm:conditional} and $\epsilon'>0$ is any positive number independent of $\eta$.
\end{theorem}

Now we want to determine the large degree $d_{max}$ and the smallest degree $d_{min}$ for the bipartite expander.
Note for edge $e$, the number of visiting random walks is equal to the degree of edge $e$'s corresponding ``edge" node in the bipartite graph. Theorem \ref{thm:sing_col_concentration} bounds $d_{max}$ and $d_{min}$.
\begin{theorem}
\label{thm:sing_col_concentration}
Choose the random walk parameters appropriately. Then the probability that a random walk visits a certain edge $e$ will be between
\begin{equation*}
P_{min}=1-(1-\frac{1}{|E|}+\frac{2\delta}{D})^{\lfloor\frac{t}{T(n)}\rfloor},
\end{equation*} and
\begin{equation*}
 P_{max}=T(n)\left(1-(1-\frac{1}{|E|}-\frac{2\delta}{D})^{\lceil\frac{t}{T(n)}\rceil}\right).
\end{equation*}
For an arbitrary $\epsilon'>0$, with probability $1-o(1)$, the number of nonzero elements in every columns of $A$ is between $(1-\epsilon') P_{min} m$ and $(1+\epsilon') P_{max}m$, when we take $m=O(T(n)k\log(n))$ random walks.
\end{theorem}

\begin{proof}
First, we establish the lower bound. We focus on the time slots starting with time index $0$, $T(n)$, $2T(n)$,.... By the definition of mixing time, the probability that this sampled walk does not visit edge $e$ is upper bounded by $(1-\frac{1}{|E|}+\frac{2\delta}{D})^{\lfloor \frac{t}{T(n)}\rfloor}$, so we have a corresponding lower bound $1-(1-\frac{1}{|E|}+\frac{2\delta}{D})^{\lfloor\frac{t}{T(n)}\rfloor}$.

For the upper bound, we consider $T(n)$ chains $f_1, f_2, ..., f_{T(n)}$of period $T(n)$.  Then for each chain, the probability that chain does not visit the edge $e$ will be lower bounded by $(1-\frac{1}{|E|}-\frac{2\delta}{D})^{\lceil\frac{t}{T(n)}\rceil}$, and so the probability that the sampled walk visits edge $e$ will be upper bounded by
\begin{equation*}
1-(1-\frac{1}{|E|}-\frac{2\delta}{D})^{\lceil\frac{t}{T(n)}\rceil}.
\end{equation*}
By a union bound over the $T(n)$ chains, the probability that the random walk ever visits edge $e$ is upper bounded by
\begin{equation*}
P_{max}=T(n)\left(1-(1-\frac{1}{|E|}-\frac{2\delta}{D})^{\lceil \frac{t}{T(n)}\rceil}\right).
\end{equation*}
When we take $t=O(\frac{|E|}{k})$, the lower bound and upper bound scale as $O(\frac{1}{kT(n)})$ and $O(\frac{1}{k})$ respectively. So by similar Chernoff bound arguments as in Theorem \ref{thm:ksetconcentration}, if $m=O(T(n)k\log(n))$, with high probability, simultaneously for all the columns, the number of non-all-zero elements concentrate between $O(\log(n))$ and $O(T(n)\log(n))$ respectively.
\end{proof}

\begin{lemma}
\label{lem:maxdegreeof2}
Any walk $W$ taken over an undirected graph can be converted to a walk that visits the same set of edges and visits each edge no more than twice.
\end{lemma}

\begin{proof} We induct on the number of nodes that the random walk visits. Apparently, for up to $2$ nodes, this claim is true. We assume this claim is true for any walk that visits up to $n$ nodes. If a random walk visits $(n+1)$ nodes, there must be a node $N$ such that when $N$ is deleted from the walk, the remaining parts of the walk remain connected. In fact, take an arbitrary node $i$ on the random walk, then all the other nodes are on a spanning tree whose root is node $i$. Then any leaf node of this tree can be deleted while all the remaining nodes remain connected. By the induction assumption, we know there exists a walk $W'$ that visits each edge of the remaining $n$-node graph for at least once but for at most twice.  Then we can construct another walk $W''$ over the $(n+1)$-node network in the following way. We start on walk $W'$. When walk $W'$ visits a node $j$ that is connected to node $N$ through an edge $e_1$ in the walk $W$, we will divert from node $j$ via edge $e_1$ to visit node $N$ and come back along the same edge to node $j$.  From there, we continue in a similar fashion along the walk $W'$ to complete constructing the new walk $W''$, which visits every edge of $W$, but no more than twice.
\end{proof}

\begin{theorem}
\label{thm:l1works}
With probability $1-o(1)$, $\ell_1$ minimization can recover any $\Theta(k)$-sparse edge vector measured using matrix $A$ generated from $m=O(T(n)k\log(n))$ independent ``good start'' regularized random walks of length $t=O(\frac{|E|}{k})$.
\end{theorem}
\begin{proof} $\ell_1$ minimization recovers every $k'$-sparse vector if and only if every nonzero vector $\w \in \mathcal{N}(A)$, $\|\w_{K'}\|_{1} \leq \alpha \|\w\|_1 $ for any edge index set $K'$ with cardinality $\Theta(k)$, where $\alpha<\frac{1}{2}$. By Theorem \ref{thm:walkexpansionforsmallk_1}, the measurement matrix $A$ generated by $m=O(T(n)k\log(n))$ ``good start'' random walks of length $O(\frac{|E|}{k})$ corresponds to a bipartite  $(k,\epsilon)$ expander graph with high probability, where $\epsilon>0$ is a constant which can be made arbitrarily close to $0$ if we choose $t$ and $m$ appropriately. Now we show for such an $A$ with expansion, the null space requirement for $\ell_1$ success is satisfied for $|K'|=\Theta(k)$. The proof in this lemma follows the same line of reasoning as in \cite{Indyk}, except for taking care of the irregularities in uneven nonzero elements in $A$ and unequal degrees for left-hand side nodes. Thus the readers are encouraged to see \cite{Indyk} for more detailed explanations.

Let $K'$ be the index set of largest elements (in amplitude) in a nonzero vector $\w \in \mathcal{N}(A)$, with cardinality $|K'|=k'\leq \frac{k}{2}$. So they correspond to $k'$ ``edge'' nodes in the bipartite graph representation for $A$. We first argue that
\begin{equation}
\label{eq:isometry}
\|A_{K'}\w_{K'}\|_1 \geq (d_{min}-4d_{max}\epsilon) \|\w_{K'}\|_1.
\end{equation}
Let us imagine a bipartite graph for $A$, but with no links (between the lefthand nodes and righthand nodes) yet. Consider the following process of adding the links  to the left-hand ``edge'' node set $K'$ one by one. We start by adding the links to the lefthand ``edge'' node that corresponds to the largest element of $\w$ in amplitude, then the links corresponding to the second largest element of $\w$ in amplitude and so on. If a newly added link is connected to a righthand side ``measurement'' node that is already ``plugged in'' by some previously added links, we will call a ``collision" occurs. If there were no ``collisions'' occurring, $\|A_{K'}\w_{K'}\|_1$ will be at least $d_{min}\|\w_{K'}\|_1$. By the expansion property of the bipartite graph, when we are done adding the links of the left hand node corresponding to the $i$-th ($i \leq k$) largest element of $\w$ in amplitude, at most $\epsilon d_{max} i$ collisions occur. Since we already rank the elements of $\w$ in amplitude, by the triangular inequality of $\ell_1$ norm, these collisions will add up to at most $2\epsilon d_{max} \|\w_{K'}\|_1$ (the term $2$ comes from the fact that the elements in $A$ are upper bounded by $2$ via regularized random walks).This will result in a loss of at most $4\epsilon d_{max} \|\w_{K'}\|_1$ in $\|A_{K'}\w_{K'}\|_1$ by the triangular inequality for $\ell_1$ norm, which leads to (\ref{eq:isometry}) .

Now we partition the index set $\{1,2,...,|E|\}$ into $l$ subsets of size $k'$ (except for the last subsect) in an decreasing order of $\w$ (in amplitudes), where $l=\lceil\frac{|E|}{k'}\rceil$. Since $A\w=0$, \emph{over the set $N(S)$ of righthand ``measurement'' nodes that are connected to $K'$}, ($K_{0}\doteq K'$),
\begin{eqnarray*}
0&=& \|A_{K'}\w_{K'}+A_{K_1}\w_{K_1}+...+A_{K_{l}}\w_{l-1}\|_{1}\\
  &\geq & (d_{min}-4d_{max}\epsilon)\|\w_{K'}\|_1-4 \epsilon d_{max} k'\sum_{d=1}^{l-1} \frac{\|\w_{K_{d-1}}\|_1}{k'}\\
  &\geq & (d_{min}-4d_{max}\epsilon)\|\w_{K'}\|_1-4 \epsilon d_{max} \|\w\|_1,
\end{eqnarray*}
where the first inequality is due to the $(k,\epsilon)$ expansion property, (which results in at most $2d_{max}k'$ link ``collisions'' between any set $K_{l}$ and $K'$ ) and the upper bound $2$ for elements in $A$. Again, please refer to \cite{Indyk} for more explanations.
So in summary, for any nonzero $\w \in \mathcal{N}(A)$,
\begin{equation*}
\|\w_{K'}\|_1 \leq  \frac{4 \frac{ \epsilon d_{max}}{d_{min}}}{1-\frac{4\epsilon d_{max}}{d_{min}}} \|\w\|_1.
\end{equation*}

As long as $\frac{\epsilon d_{max}}{d_{min}}<\frac{1}{12}$, $\ell_1$ minimization can recover up to any $k'$-sparse signal via $O(T(n)k\log(n))$ measurements, where $k' \leq \frac{k}{2}$ (conditioned on expansion property for $A$ by setting $t$ and $m$ appropriately, which is possible from Theorem \ref{thm:walkexpansionforsmallk_1} and \ref{thm:sing_col_concentration}).
\end{proof}

%
%

\section{Numerical Examples}
\label{sec:numerical}
 In this section, we will provide numerical simulation results demonstrating the performance of compressive sensing over graphs. In all the simulations, we generate the the measurement matrix $A$ from independent random walks of certain lengths, subject to the graph topology constraints. \\

 \noindent \textbf{Example 1}
 Figure \ref{fig:n50complete} shows the recovery percentage of $\ell_1$ minimization for $k$-sparse edge signal vector over a complete graph with $50$ vertices and $1225$ edges. The $k$ edges with nonzero elements are uniform randomly chosen among the $1225$ edges. For this example, we take $m=612$ random walks of length $t=612$ to collect $612$ measurements. Two scenarios are considered. One is for the edge signal vectors with real-numbered nonzero Gaussian distributed elements, which can take positive and negative values. The other scenario is for vectors with nonnegative nonzero elements, for which we impose the nonnegative constraints in $\ell_1$ minimization decoding. Saving $50$ percent of measurements, $\ell_1$ minimization can recover real-numbered sparse vectors with $17$ percent nonzero elements or nonnegative sparse vectors with about $24$ percent nonzero elements, even under the graph constraints.\\


\begin{figure}
  \centering
  \includegraphics[width=0.40\textwidth]{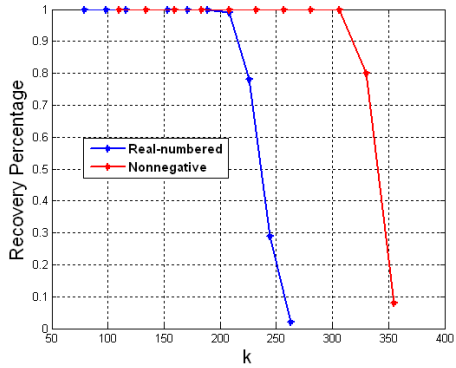}
   \caption{$n=50$ Complete Graph, with $t=612$ and $m=612$}
  \label{fig:n50complete}
\end{figure}
\noindent \textbf{Example 2}
In this example, we consider a random graph model of $50$ nodes, where there is an edge with probability $p=0.5$ between any two nodes. So on average, we have around 600 edges in the final graph. We tested $\ell_1$ decoding for real-numbered sparse signal recoveries in the same fashion as in Example 1. The length $t$ of each random walk is set as one third of $|E|$. In Figure \ref{fig:n50p05}, we plot the relationship between the number of measurements and the maximum recoverable sparsity $k$. A sparsity is deemed recoverable if 99 percent of $k$-sparse vectors have been recovered in the experiment.


\begin{figure}
  \centering
  \includegraphics[width=0.40\textwidth]{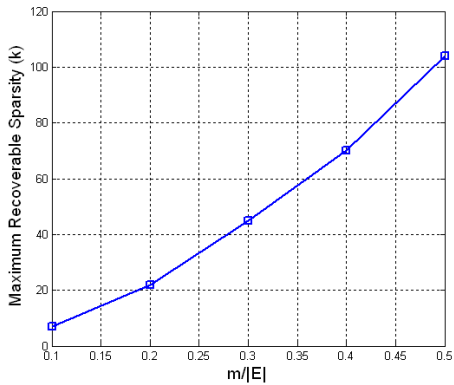}
   \caption{$n=50$ Random Graph}
  \label{fig:n50p05}
\end{figure}

\section{Conclusion}
\label{sec:conclusion}
 We study network tomography problems from the angle of compressive sensing. The unknown vectors to be recovered are sparse vectors representing certain parameters of the links over the graph. The collective additive measurements we are allowed to take must follow paths over the underlying graphs. For a sufficiently connected graph with $n$ node, we find that $O(k \log(n))$ path measurements are enough to recover any sparse link vector with no more than $k$ nonzero elements. We further demonstrate that $\ell_1$ minimization can be used to recover such sparse vectors here with theoretical guarantee. Further research is needed to find efficient ways to construct measurement paths. In addition, it is also of interest to investigate the possibility of using nonlinear measurements and low-rank matrix recovery \cite{yihongwu}\cite{zrwq}. So far we have only studied compressive sensing over graphs for ideally sparse signals and extensions to noisy measurements are part of future work. It is also interesting to consider more efficient polynomial-time algorithms for compressive sensing over graphs \cite{XuHassibi}.

\bibliographystyle{IEEEbib}

\end{document}